\documentclass[12pt,amsmath,amssymb]{iopart}

\usepackage{iopams}
\begin{document}

\title[Construnctions of LOCC indistinguishable set]
{Construnctions of LOCC indistinguishable set of generalized Bell states}

\author{Jiang-Tao Yuan, Cai-Hong Wang, Ying-Hui Yang \& Shi-Jiao Geng}

\address{School of Mathematics and Information Science, Henan Polytechnic University, Jiaozuo, 454000, China}
\ead{jtyuan@hpu.edu.cn, chwang@hpu.edu.cn, yangyinghui4149@163.com}
\vspace{10pt}
\begin{indented}
\item[]May 2018
\end{indented}

\begin{abstract}
In this paper, we mainly consider the local indistinguishability of the set of  mutually orthogonal bipartite generalized
Bell states (GBSs). We construct small sets of GBSs with  cardinality smaller than $d$ which are not
distinguished by one-way local operations and classical communication (1-LOCC) in $d\otimes d$.
The constructions, based on linear system and Vandermonde matrix, is simple and effective.
The results give a unified upper bound for the minimum cardinality of 1-LOCC indistinguishable set of GBSs,
and greatly improve previous results in [Zhang \emph{et al.}, Phys. Rev. A 91, 012329 (2015); Wang \emph{et al.}, Quantum Inf. Process. 15, 1661 (2016)].
The case that $d$ is odd of  the results also shows that
the set of 4 GBSs in $5\otimes 5$ in [Fan, Phys. Rev. A 75, 014305 (2007)]
is indeed a 1-LOCC indistinguishable set which can not be distinguished by Fan's method.
\end{abstract}

\noindent{\it Keywords\/}: maximally entangled states, generalized Bell states, local indistinguishability

%
%
%
%
%

\section{Introduction}
In quantum mechanics, any set of orthogonal states can be discriminated.
In general, for bipartite systems,  local operations and classical communication (LOCC)
is not sufficient to distinguish among orthogonal states \cite{benn1999pra,walg2000prl,walg2002prl,gho2001prl}.
Any two orthogonal states can be perfectly distinguished by LOCC,
a complete orthogonal basis of maximally  entangled states (MESs) is locally indistinguishable,
deterministically or probabilistically \cite{horo2003prl,fan2004prl,fan2007pra}.
The nonlocal nature of quantum information is revealed when a set of orthogonal states
of a composite quantum system cannot be perfectly distinguished by LOCC.
This has been very useful in exploring quantum nonlocality and its
relationship with entanglement \cite{benn1999pra,walg2002prl,horo2003prl,benn1999prl}.

It is known that $d+1$ or more MESs in $d\otimes d$ are not perfectly locally distinguishable \cite{fan2004prl,fan2007pra,gho2004pra}.
Therefore, it is natural to ask whether a set of $N\leq d$ orthogonal MESs in $d\otimes d$
can be perfectly distinguished by LOCC for $d\geq4$  \cite{nath2005jmp}.
Firstly, Ghosh et al.  \cite{gho2004pra} showed examples of sets of $d$ generalized Bell states (GBSs) in
$d\otimes d$ for $d = 4,5$ that cannot be perfectly distinguished by one-way LOCC (1-LOCC)  \cite{nath2013pra}.
Fan \cite{fan2007pra} provided an example of $4$ GBSs in $5\otimes 5$ that potentially cannot be distinguished by LOCC
since Fan's method \cite{fan2004prl} does not work for the example.
Bandyopadhyay et al. \cite{band2011njp} gave examples of sets of $d-1$ GBSs in
$d\otimes d$ for $d =5,6$ that cannot be distinguished by 1-LOCC.
Yu et al. \cite{yu2012prl} constructed a set of 4 MESs in $4\otimes 4$ which is indistinguishable by positive partial transpose (PPT) operations.

Recently, Zhang et al. \cite{zhang2015pra} defined the function $f(d)$ which is the
minimum cardinality of 1-LOCC indistinguishable set of MESs in $d\otimes d$,
and proved that $f(d)\leq \lceil\frac{d+4}{2}\rceil$.
Wang \emph{et al.} \cite{wang2016qip} constructed a 1-LOCC indistinguishable set of $3\lceil \sqrt{d}\rceil-1$ GBSs,
proved that $3$ GBSs in $d\otimes d\ (d\geq 4)$ are always LOCC distinguishable \cite{wang2017qip},
and there exist a 1-LOCC indistinguishable set of $4$ MESs in $d\otimes d\ (d\geq 4)$ \cite{wang2016qip,nath2013pra}.

Obviously, if $d\geq 4$, then $f(d)\leq 4$.
Meanwhile Fan \cite{fan2004prl} showed that if $d$ is prime
then a set of $l$ GBSs satisfying $\frac{l(l-1)}{2}\leq d$ in $d\otimes d$ is LOCC distinguishable.
Hence the function $f(d)$ can not describe Fan's result.
It is natural to define the function $f_{GBS}(d)$ which is the
minimum cardinality of 1-LOCC indistinguishable set of GBSs in $d\otimes d$.
By the result in  \cite{zhang2015pra,wang2016qip,zhang2014qip}, $f_{GBS}(d)\leq \min\{\lceil\frac{d+4}{2}\rceil,3\lceil \sqrt{d}\rceil-1\}$.
It is known that  $f_{GBS}(2)=3$, $f_{GBS}(3)=4$, $f_{GBS}(4)=4$ and $f_{GBS}(5)=4$
\cite{walg2000prl,nath2005jmp,gho2004pra,band2011njp,wang2017qip,tian2016pra,sing2017pra}.
In general, for $d\geq 5$, $f_{GBS}(d)\leq \min\{\lceil\frac{d+4}{2}\rceil,3\lceil \sqrt{d}\rceil-1\}$
which often is loose and not a exact value of $f_{GBS}(d)$.
For example, $f_{GBS}(5)=4$ \cite{fan2004prl,band2011njp,wang2017qip}  and $\min\{\lceil\frac{5+4}{2}\rceil,3\lceil \sqrt{5}\rceil-1\}=5$.
Till now, though there is a set of 8 mutually unbiased bases (MUB) in the space $\mathbb{C}^{7}$ \cite{fan2004prl,band2002alg},
the exact value of $f_{GBS}(7)$ is unknown.

In this paper, we focus on constructing the general 1-LOCC indistinguishable set of GBSs in $d\otimes d$.
When $d$ is odd, we show that there exist 1-LOCC indistinguishable sets of $m_{odd}$ GBSs
where $m_{odd}$ is not more than
\begin{eqnarray}
\min\{\frac{d+3}{2},\lfloor\frac{d+1}{4}\rfloor+5,2\lceil \sqrt{d}\rceil+\lceil\frac{\lceil \frac{d-1}{4}\rceil}{\lceil \sqrt{d}\rceil}\rceil\},
\end{eqnarray}
so $f_{GBS}(d)\leq m_{odd}.$
In particular,  $f_{GBS}(7)\leq  \frac{7+3}{2}=5$ and this together with Fan's result \cite{fan2004prl} imply $f_{GBS}(7)=5$.
When $d$ is even, we construct 1-LOCC indistinguishable sets of $m_{even}$ GBSs
where $m_{even}$ is not more than
\begin{eqnarray}
\min\{\lceil\frac{d}{4}\rceil+3,2\lceil \sqrt{\frac{d+2}{2}}\rceil+\lceil\frac{\lceil \frac{d+1}{4}\rceil}{\lceil \sqrt{\frac{d+2}{2}}\rceil}\rceil\},
\end{eqnarray}
so $f_{GBS}(d)\leq m_{even}$, $f_{GBS}(4)\leq 4$ \cite{wang2017qip,tian2016pra,sing2017pra}, $f_{GBS}(6)\leq 5$ and $f_{GBS}(8)\leq 5$.
For 1-LOCC indistinguishability of the states,
a simple and effective  method is presented which
is based on linear system and Vandermonde matrix.
Our results imply that  the set of 4 GBSs in $5\otimes 5$ in \cite{fan2007pra} is really a 1-LOCC indistinguishable set,
this is an interesting example since Fan's method \cite{fan2004prl} does not work for this case.

The rest of this paper is organized as follows.
In Sec. II, we recall some relevant notions and results.
In Sec. III, the local indistinguishability of GBSs is discussed for the case that $d$ is odd.
In Sec. IV, the local indistinguishability of GBSs is considered for the case that $d$ is even.
At the end, in Sec. IV, we draw the conclusion.

\section{Preliminaries}
\newtheorem{definition}{\indent Definition}
\newtheorem{lemma}{\indent Lemma}
\newtheorem{theorem}{\indent Theorem}
\newtheorem{corollary}{\indent Corollary}

\def\QEDclosed{\mbox{\rule[0pt]{1.3ex}{1.3ex}}}
\def\QED{\QEDclosed}
\def\proof{\indent{\em Proof}.}
\def\endproof{\hspace*{\fill}~\QED\par\endtrivlist\unskip}

Consider a Hilbert space with dimension $d$, $\{|j\rangle\}_{j=0}^{d-1}$ is the computational basis.
Let $U_{m,n}=X^{m}Z^{n}, m, n=0,1,\ldots,d-1$ be generalized Pauli matrices
constituting a basis of unitary operators,
and $X|j\rangle=|j+1$ mod $d\rangle$, $Z|j\rangle=w^{j}|j\rangle$, $w=e^{2\pi i/d}$.

In a quantum system $\mathcal {H}_{A}\otimes\mathcal {H}_{B}$ of dimension $d\otimes d$,
the canonical maximally entangled state $|\Phi\rangle$ in $d\otimes d$  is $|\Phi_{00}\rangle=(1/\sqrt{d})\sum_{j=0}^{d-1}|jj\rangle$.
We know that $(I\otimes U)|\Phi\rangle=(U^{T}\otimes I)|\Phi\rangle$, where $T$ means matrix transposition.
Any maximally entangled state can be written as $|\Psi\rangle=(I\otimes U)|\Phi\rangle$ where $U$ is unitary.
If $U=X^{m}Z^{n}$, the states
\begin{eqnarray}
|\Phi_{m,n}\rangle=(I\otimes U_{m,n})|\Phi\rangle
\end{eqnarray}
are called generalized Bell states.
For simplicity, denote $U_{m,n}=X^{m}Z^{n}\doteq(m,n)$,
and $\{(I\otimes U_{m,n})|\Phi\rangle\}\doteq\{U_{m,n}\}\doteq\{X^{m}Z^{n}\}\doteq\{(m,n)\}$.

For a set of GBSs $S^{d}=\{(m_{j},n_{j})\}_{j=1}^{l}$, the corresponding pairwise difference set $\Delta U$ means
\begin{eqnarray*}
\Delta U=\{(m_{jk}, n_{jk})|m_{jk}=m_{j}-m_{k}, n_{jk}=n_{j}-n_{k}, j\neq k\}.
\end{eqnarray*}

\begin{lemma}[\cite{zhang2014qip}]\label{zhang2014lem}
In $d\otimes d$, a set $\{|\Phi_{m_{j}n_{j}}\rangle\}_{j=1}^{l}$ of $l$ GBSs ($l\leq d$)
can be perfectly distinguished by 1-LOCC if and only if
there exists at least one state $|\alpha\rangle$ for which
the set $\{U_{m_{j}n_{j}}|\alpha\rangle\}_{j=1}^{l}$ are pairwise orthogonal.
\end{lemma}
\begin{lemma}\label{lem2.1}
Let $W_{d\times d}=[w^{ij}]_{i,j=0}^{d-1}=[|w_{j}\rangle]_{j=0}^{d-1}$ be a Vandermonde matrix
where $|w_{j}\rangle$ is a column vector.
If  $k\leq d,\ 0\leq i_{1}\leq d-1,\ 0\leq j_{1}<\cdots<j_{k}\leq d-1$,
then each submarix
$$W\left(
\begin{array}{llll}
i_{1}&i_{1}+1&\cdots&i_{1}+k-1\\
j_{1}&j_{2}&\cdots&j_{k}
\end{array}
\right)=
\left[
\begin{array}{llll}
 w^{i_{1}j_{1}}&w^{i_{1}j_{2}} & \cdots & w^{i_{1}j_{k}}\\
w^{(i_{1}+1)j_{1}}&w^{(i_{1}+1)j_{2}} & \cdots & w^{(i_{1}+1)j_{k}}\\
\cdots& \cdots & \cdots & \cdots\\
w^{(i_{1}+k-1)j_{1}}&w^{(i_{1}+k-1)j_{2}} & \cdots & w^{(i_{1}+k-1)j_{k}}\\
\end{array}
\right]$$ of $W_{d\times d}$ is invertible.
\end{lemma}
\begin{proof}
If $k=d$, since $w=e^{2\pi i/d}$, $\det(W_{d\times d})=\prod_{0\leq i<j\leq d-1}(w^{j}-w^{i})\neq0$.
If $k<d$, we have
$\det(W\left(
\begin{array}{llll}
i_{1}&\cdots&i_{1}+k-1\\
j_{1}&\cdots&j_{k}\end{array}
\right))
=w^{i_{1}(j_{1}+\cdots+j_{k})}\prod_{1\leq l<m\leq k}(w^{j_{m}}-w^{j_{l}})\neq0$.
\end{proof}

\section{Constructions of 1-LOCC indistinguishable set when $d$ is odd}
In this section,assume that  $d$ is odd and $d\geq 5$ since the case $d< 5$ is known.
We present our families of examples with different methods.
The first family is inspired by Fan's example that there is a potentially LOCC indistinguishable set of 4 GBSs in $5\otimes 5$ \cite{fan2007pra}.
The second family improves the results in Zhang et al. \cite{zhang2015pra}.
The last family is a generalization of the results in Wang et al. \cite{wang2016qip}.

\begin{lemma}\label{lem3.1}
Let $S^{d}=\{(m_{j},n_{j})\}_{j=1}^{l}$ be a set of GBSs, $I_{1}\doteq\{(1,i)\}_{i=0}^{d-1}$.
If there exists $i_{0}$ such that $0< i_{0}< d-1$ and $\Delta U\supseteq I_{1}\cup I_{0}$
where $I_{0}=\{(0,i)\}_{i=i_{0}}^{i_{0}+\lfloor\frac{d}{2}\rfloor-1}$,
then $S^{d}$ is 1-LOCC indistinguishable.
In particular, the assertion holds if the set $I_{0}$ is replaced with
$\{(0,i)\}_{i=\lfloor\frac{d}{2}\rfloor-\lceil\frac{d-1}{4}\rceil+1}^{\lfloor\frac{d}{2}\rfloor}$.
\end{lemma}

\begin{proof}
Suppose that $S^{d}$ can be distinguished by 1-LOCC,
then by Lemma \ref{zhang2014lem}, there exists a normalized vector $|\alpha\rangle=\sum_{j=0}^{d-1}\alpha_{j}|j\rangle$
such that $\{U_{m_{j}n_{j}}|\alpha\rangle:\{(m_{j},n_{j})\}\subset S^{d}\}$ is a orthonormal set.
It means that $\langle\alpha|U_{m_{j}n_{j}}^{\dag}U_{m_{k}n_{k}}|\alpha\rangle=0$, $j\neq k$.
So $\langle\alpha|U_{m_{kj}n_{kj}}|\alpha\rangle=0$.
Hence $\Delta U\supseteq I_{1}$ implies
$\langle\alpha|U_{1i}|\alpha\rangle=0$,
i.e., $\langle(\alpha_{j}^{*}\alpha_{j+1})_{j=0}^{d-1}|w_{i}\rangle=0$
where $(\alpha_{j}^{*}\alpha_{j+1})_{j=0}^{d-1}=(\alpha_{0}^{*}\alpha_{1},\alpha_{1}^{*}\alpha_{2},\ldots,\alpha_{d-1}^{*}\alpha_{0})$.
Therefore $(\alpha_{j}^{*}\alpha_{j+1})_{j=0}^{d-1}$ is a zero vector since $\{|w_{i}\rangle\}_{i=0}^{d-1}$ is a base.
This ensures that there is at least $\lceil\frac{d}{2}\rceil$ $\alpha_{j}=0$,
we can assume
\begin{eqnarray}\label{eq1lem3.1}
\alpha_{c_{1}}\neq0,\cdots,\alpha_{c_{n}\neq0}, n\leq \lfloor\frac{d}{2}\rfloor.
\end{eqnarray}

On the other hand, denote $b_{i}\doteq \langle w_{i}|(|\alpha_{j}|^{2})_{j=0}^{d-1}\rangle$
the Fourier coefficient of the vector $(|\alpha_{j}|^{2})_{j=0}^{d-1}$ with respect to the base $\{|w_{i}\rangle\}_{i=0}^{d-1}$.
Then $b_{0}=1$ ($|\alpha\rangle$ is a normalized vector) and
$|(|\alpha_{j}|^{2})_{j=0}^{d-1}\rangle=W_{d\times d}|(b_{j})_{j=0}^{d-1}\rangle$,
i.e.,
\begin{eqnarray}\label{eq2lem3.1}
W^{\ast}_{d\times d}|(|\alpha_{j}|^{2})_{j=0}^{d-1}\rangle=|(b_{j})_{j=0}^{d-1}\rangle.
\end{eqnarray}

Similarly, $\Delta U\supseteq I_{0}$ implies
$\langle\alpha|U_{0i}|\alpha\rangle=0$,
i.e.,
\begin{eqnarray}\label{eq3lem3.1}
\langle(|\alpha_{j}|^{2})_{j=0}^{d-1}|w_{i}\rangle=0=b_{i},\
i=i_{0},\cdots,i_{0}+\lfloor\frac{d}{2}\rfloor-1.
\end{eqnarray}
It follows from (\ref{eq1lem3.1})-(\ref{eq3lem3.1}) that
\begin{eqnarray}\label{eq4lem3.1}
W^{\ast}\left(
\begin{array}{llll}
i_{0}&\cdots&i_{0}+\lfloor\frac{d}{2}\rfloor-1\\
c_{1}&\cdots&c_{n}
\end{array}
\right)
|(|\alpha_{c_{j}}|^{2})_{j=1}^{n}\rangle=|0\rangle.
\end{eqnarray}
This means the linear system $W^{\ast}\left(
\begin{array}{llll}
i_{0}&\cdots&i_{0}+\lfloor\frac{d}{2}\rfloor-1\\
c_{1}&\cdots&c_{n}
\end{array}
\right)X=|0\rangle$ has a nonzero solution $|(|\alpha_{c_{j}}|^{2})_{j=1}^{n}\rangle$.
By Lemma \ref{lem2.1}, rank$(W^{\ast}\left(
\begin{array}{llll}
i_{0}&\cdots&i_{0}+\lfloor\frac{d}{2}\rfloor-1\\
c_{1}&\cdots&c_{n}
\end{array}
\right))=n\leq \lfloor\frac{d}{2}\rfloor$.
This is a contradiction.

In particular, if the set $I_{0}$ is replaced with
$\{(0,i)\}_{i=\lfloor\frac{d}{2}\rfloor-\lceil\frac{d-1}{4}\rceil+1}^{\lfloor\frac{d}{2}\rfloor}$.
Then
$\langle\alpha|U_{0i}|\alpha\rangle=0$
which is equivalent to
$\langle\alpha|U_{0,d-i}|\alpha\rangle=0$,
i.e.,
\begin{eqnarray}\label{eq5lem3.1}\nonumber
\langle(|\alpha_{j}|^{2})_{j=0}^{d-1}|w_{i}\rangle=0=\langle(|\alpha_{j}|^{2})_{j=0}^{d-1}|w_{d-i}\rangle.
\end{eqnarray}

So $\Delta U\supseteq\{(0,i)\}_{i=\lfloor\frac{d}{2}\rfloor-\lceil\frac{d-1}{4}\rceil+1}^{\lceil\frac{d}{2}\rceil+\lceil\frac{d-1}{4}\rceil-1}$
and the assertion follows by $2\lceil\frac{d-1}{4}\rceil\geq \lfloor\frac{d}{2}\rfloor=|I_{0}|$.
\end{proof}

\subsection{LOCC indistinguishable set of $\frac{d+3}{2}$ GBSs}

At the end of \cite{fan2007pra}, Fan gave a set of GBSs
$S^{5}=\{(0,0),(2,0),(1,1),(1,3)\}$ with particular interest
because the set can not be LOCC distinguished by Fan's method in \cite{fan2004prl}
and thus it potentially can not be LOCC distinguished.
In this subsection, it is shown that there exists a LOCC indistinguishable set of $\frac{d+3}{2}$ GBSs
which implies that Fan's set of 4 GBSs in $5\otimes 5$ is indeed a 1-LOCC indistinguishable set.

\begin{theorem}\label{th3.1}
Let $S^{d}=\{(0,0),(2,0),(1,2i-1),i=1,\cdots,\frac{d-1}{2}\}$ be a set of GBSs,
then $S^{d}$ is 1-LOCC indistinguishable and $f_{GBS}(d)\leq \frac{d+3}{2}.$
Especially, $f_{GBS}(5)=4$, $f_{GBS}(7)=5$.
\end{theorem}

When $d$ is odd, the cardinality of $S^{d}$ ($|S^{d}|=\frac{d+3}{2}$) in Theorem \ref{th3.1} is smaller than
the result in Zhang et al. \cite{zhang2015pra} ($|S^{d}|=\frac{d+5}{2}$).

\begin{proof}
Suppose that $S^{d}$ can be distinguished by 1-LOCC,
according to Lemma \ref{zhang2014lem}, there exists a normalized vector $|\alpha\rangle=\sum_{j=0}^{d-1}\alpha_{j}|j\rangle$
such that $\{U_{m_{j}n_{j}}|\alpha\rangle:\{(m_{j},n_{j})\}\subset S^{d}\}$ is a orthonormal set.
It means that $\langle\alpha|U_{m_{kj}n_{kj}}|\alpha\rangle=0$, $j\neq k$.
By assumption, $\Delta U\supseteq\{(0,i)\}_{i=2}^{d-2}\cup\{(1,i)\}_{i=1}^{d-1}\cup\{(2,0)\}$.
We will show that $\Delta U\supseteq\{(1,0)\}$.
In fact, by $\Delta U\supseteq\{(1,i)\}_{i=1}^{d-1}$, i.e.,
$\langle\alpha|U_{1i}|\alpha\rangle=0$,
we have $|(\alpha_{j}^{*}\alpha_{j+1})_{j=0}^{d-1}\rangle=\lambda |w_{0}\rangle=\lambda (1,\cdots,1)^{T}$.
If $\lambda\neq 0$, then $\alpha_{j}\neq 0$,
$\alpha_{j}^{*}\alpha_{j+2}=\frac{\alpha_{j}^{*}\alpha_{j+1}\alpha_{j+1}^{*}\alpha_{j+2}}{|\alpha_{j+1}|^{2}}
=\frac{\lambda^{2}}{|\alpha_{j+1}|^{2}}$.
So $|(\alpha_{j}^{*}\alpha_{j+2})_{j=0}^{d-1}\rangle=\lambda^{2}|(|\alpha_{j+1}|^{2})_{j=0}^{d-1}\rangle$
and it contradicts with $\Delta U\supseteq\{(2,0)\}$.
Therefore $\lambda= 0$, $\Delta U\supseteq\{(1,0)\}$.
By Lemma \ref{lem3.1}, $S^{d}$ is indistinguished by 1-LOCC and this is a contradiction.
\end{proof}

\subsection{LOCC indistinguishable set of $5+\lfloor\frac{d+1}{4}\rfloor$ GBSs}
In this subsection, assume that  $d\geq 9$ since the case $d\leq 7$ is known.

\begin{theorem}\label{th3.2}
Let $S^{d}=\{(1,2i-1)\}_{i=1}^{\lfloor\frac{d+1}{4}\rfloor}\cup\{(0,0),(1,0),(1,1),(1,\lfloor\frac{d}{2}\rfloor),(1,\lceil\frac{d}{2}\rceil)\}
$ be a set of GBSs,
then $S^{d}$ is 1-LOCC indistinguishable and $f_{GBS}(d)\leq 5+\lfloor\frac{d+1}{4}\rfloor.$
\end{theorem}

When $d$ is odd and $d\geq 9$, the cardinality of $S^{d}$ ($|S^{d}|=5+\lfloor\frac{d+1}{4}\rfloor$) in Theorem \ref{th3.2}
is smaller than the result in Zhang et al. \cite{zhang2015pra} ($|S^{d}|=\frac{d+5}{2}$).
When $d\geq 13$, $|S^{d}|$ in Theorem \ref{th3.2} is not more than $|S^{d}|$ in Theorem \ref{th3.1}.
See Table \ref{tab3.1} for comparison of $|S^{d}|$ in \cite{zhang2015pra} and Theorem \ref{th3.2}.

\begin{proof}
It is easy to check that
$\Delta U\supseteq\{(0,i)\}_{i=\frac{d-1}{2}-\lceil\frac{d-1}{4}\rceil+1}^{\frac{d-1}{2}}\cup\{(1,i)\}_{i=0}^{d-1}$.
Hence the assertion follows by Lemma \ref{lem3.1}.
\end{proof}

\subsection{LOCC indistinguishable set of no more than
$2\lceil \sqrt{d}\rceil+\lceil\frac{\lceil \frac{d-1}{4}\rceil}{\lceil \sqrt{d}\rceil}\rceil$ GBSs}
In this subsection, assume that  $d\geq 9$.

Wang et al. \cite{wang2016qip} constructed a set of $3\lceil \sqrt{d}\rceil-1$ GBSs
Here we construct a small set of GBSs with $|S^{d}|\leq 2\lceil \sqrt{d}\rceil+\lceil\frac{\lceil \frac{d-1}{4}\rceil}{\lceil \sqrt{d}\rceil}\rceil$
based on Lemma \ref{lem3.2}.

\begin{lemma}\label{lem3.2}\nonumber
Let $m$ be a positive integer, $S^{d}(m)=\{(0,i)\}_{i=0}^{m-1}\cup\{(1,im-1)\}_{i=1}^{\lceil\frac{d}{m}\rceil}
\cup\{(0,\frac{d-1}{2}-im)\}_{i=0}^{\lceil\frac{\lceil\frac{d-1}{4}\rceil}{m}\rceil-1}$ a set of GBSs,
then $S^{d}(m)$ is 1-LOCC indistinguishable and
$|S^{d}(m)|=m+\lceil\frac{d}{m}\rceil+\lceil\frac{\lceil\frac{d-1}{4}\rceil}{m}\rceil.$
\end{lemma}
\begin{proof}
Since $\Delta U\supseteq\{(0,i)\}_{i=1}^{d-1}\cup\{(1,i)\}_{i=0}^{d-1}$,
the assertion follows by Lemma \ref{lem3.1}.
\end{proof}

\begin{theorem}\label{th3.3}
If $\lfloor\sqrt{d}\rfloor^{2}\leq d\leq \lfloor\sqrt{d}\rfloor\lceil\sqrt{d}\rceil$,
then there exists a 1-LOCC indistinguishable set $S^{d}$ of GBSs with $|S^{d}|=\lfloor\sqrt{d}\rfloor+\lceil\sqrt{d}\rceil+\lceil\frac{\lceil\frac{d-1}{4}\rceil}{\lceil\sqrt{d}\rceil+k_{1}}\rceil$ 
where $k_{1}=\max\{k: d\leq (\lfloor\sqrt{d}\rfloor-k)(\lceil\sqrt{d}\rceil+k)\}$;
if $\lfloor\sqrt{d}\rfloor\lceil\sqrt{d}\rceil< d\leq \lfloor\sqrt{d}\rfloor(\lceil\sqrt{d}\rceil+1)$,
then  there exists a 1-LOCC indistinguishable set $S^{d}$ of GBSs with $|S^{d}|=\lfloor\sqrt{d}\rfloor+\lceil\sqrt{d}\rceil+1+\lceil\frac{\lceil\frac{d-1}{4}\rceil}{\lceil\sqrt{d}\rceil+1+k_{2}}\rceil$
where $k_{2}=\max\{k: d\leq (\lfloor\sqrt{d}\rfloor-k)(\lceil\sqrt{d}\rceil+1+k)\}$.
\end{theorem}

\begin{proof}
If $\lfloor\sqrt{d}\rfloor^{2}\leq d\leq \lfloor\sqrt{d}\rfloor\lceil\sqrt{d}\rceil$,
by Lemma \ref{lem3.2}, the set $S^{d}(\lceil\sqrt{d}\rceil+k_{1})$
is 1-LOCC indistinguishable with $|S^{d}(\lceil\sqrt{d}\rceil+k_{1})|=\lfloor\sqrt{d}\rfloor+\lceil\sqrt{d}\rceil+\lceil\frac{\lceil\frac{d-1}{4}\rceil}{\lceil\sqrt{d}\rceil+k_{1}}\rceil.$
Similarly, when $\lfloor\sqrt{d}\rfloor\lceil\sqrt{d}\rceil< d\leq \lfloor\sqrt{d}\rfloor(\lceil\sqrt{d}\rceil+1)$,
the set $S^{d}(\lceil\sqrt{d}\rceil+1+k_{2})$ is 1-LOCC indistinguishable with
$|S^{d}(\lceil\sqrt{d}\rceil+1+k_{2})| =\lfloor\sqrt{d}\rfloor+\lceil\sqrt{d}\rceil+1+\lceil\frac{\lceil\frac{d-1}{4}\rceil}{\lceil\sqrt{d}\rceil+1+k_{2}}\rceil$.
\end{proof}
When $d\geq 35$, $|S^{d}|$ in Theorem \ref{th3.3} is not more than $|S^{d}|$ in Theorem \ref{th3.2}.
See Table \ref{tab3.2} for comparison of $|S^{d}|$ in \cite[Theorem 1]{wang2016qip} and Theorem \ref{th3.3}.

\begin{table*}
\caption{\label{tab3.1}Comparison of \cite[Theorem 1]{zhang2015pra} ($|S^{d}|=\frac{d+5}{2}$) and Theorem \ref{th3.1}-\ref{th3.2}
($|S^{d}|=\frac{d+3}{2}$, $|S^{d}|=5+\lceil\frac{d}{4}\rceil$).}
\footnotesize
\begin{tabular}{c|ccccccccccccccccccc}
\br
$d$ &5&7&9&11&13$^a$&15&17&19&29&35$^b$&39&49&59&69&79&89&99\\
\hline
$|S^{d}|$ in \cite{zhang2015pra}&5&6&7&8&9&10&11&12&17&20&22&27&32&37&42&47&52 \\
$|S^{d}|$ in Theorem \ref{th3.1}&4&5&6&7&8&9&10&11&16&19&21&26&31&36&41&46&51\\
$|S^{d}|$ in Theorem \ref{th3.2}&&&7&8&8&9&9&10&13&14&15&17&20&23&25&28&30\\
\br
\end{tabular}
$^{a}$$|S^{d}|$ in Theorem \ref{th3.2} is not more than $|S^{d}|$ in Theorem \ref{th3.1} when $d\geq 13$.
\end{table*}

\begin{table*}
\caption{\label{tab3.2}Comparison of \cite[Theorem 1]{wang2016qip} ($|S^{d}|=3\lceil \sqrt{d}\rceil-1$) and Theorem \ref{th3.3}
($|S^{d}|=\lfloor\sqrt{d}\rfloor+\lceil\sqrt{d}\rceil+\lceil\frac{\lceil\frac{d-1}{4}\rceil}{\lceil\sqrt{d}\rceil+k_{1}}\rceil$,
$|S^{d}|=\lfloor\sqrt{d}\rfloor+\lceil\sqrt{d}\rceil+1+\lceil\frac{\lceil\frac{d-1}{4}\rceil}{\lceil\sqrt{d}\rceil+1+k_{2}}\rceil$).}

\footnotesize
\begin{tabular}{c|ccccccccccccccccccc}
\br
$d$ &5&7&9&11&13&15&17&19&29&35$^b$&39&49&59&69&79&89&99\\
\hline
$|S^{d}|$ in \cite{wang2016qip}&&&8&11&11&11&14&14&17&17&20&20&23&26&26&29&29\\
$|S^{d}|$ in Theorem \ref{th3.3}&&&7&8&9&9&10&10&13&14&15&16&18&19&20&22&23\\
\br
\end{tabular}
   $^{b}$$|S^{d}|$ in Theorem \ref{th3.3} is not more than $|S^{d}|$ in Theorem \ref{th3.2} when $d\geq 35$.
\end{table*}
\normalsize

\section{Constructions of 1-LOCC indistinguishable set when $d$ is even}
In this section,assume that  $d$ is even and $d\geq 4$ since the case $d< 4$ is known.
Two families of examples are presented.
The first family is inspired by examples in Zhang et al. \cite{zhang2015pra}.
The second family improves the results in Wang et al. \cite{wang2016qip}.

\begin{lemma}\label{lem4.1}
Let $S^{d}=\{(m_{j},n_{j})\}_{j=1}^{l}$ be a set of GBSs, $I_{\frac{d}{2}}\doteq\{(\frac{d}{2},i)\}_{i=0}^{\frac{d}{2}}$.
If there exists $i_{0}$ such that $0< i_{0}< d-1$ and $\Delta U\supseteq I_{\frac{d}{2}}\cup I_{0}$
where $I_{0}=\{(0,i)\}_{i=i_{0}}^{i_{0}+\frac{d}{2}-1}$,
then $S^{d}$ is 1-LOCC indistinguishable.
In particular, the assertion holds if the set $I_{0}$ is replaced with
$\{(0,i)\}_{i=\frac{d}{2}-\lceil\frac{d+1}{4}\rceil+1}^{\frac{d}{2}}$.
\end{lemma}

\begin{proof}
Suppose that $S^{d}$ can be distinguished by 1-LOCC,
then there exists a normalized vector $|\alpha\rangle=\sum_{j=0}^{d-1}\alpha_{j}|j\rangle$
such that $\{U_{m_{j}n_{j}}|\alpha\rangle:\{(m_{j},n_{j})\}\subset S^{d}\}$ is a orthonormal set.
It means that $\langle\alpha|U_{m_{j}n_{j}}^{\dag}U_{m_{k}n_{k}}|\alpha\rangle=0$,
$\langle\alpha|U_{m_{kj}n_{kj}}|\alpha\rangle=0$, $j\neq k$.
Hence $\Delta U\supseteq I_{\frac{d}{2}}$ implies
\begin{eqnarray}\label{}\nonumber
\langle\alpha|U_{\frac{d}{2},i}|\alpha\rangle=0,i=0,\cdots,\frac{d}{2}
\end{eqnarray}
which is equivalent to
\begin{eqnarray}\label{}\nonumber
\langle\alpha|U_{\frac{d}{2},d-i}|\alpha\rangle=0,i=0,\cdots,\frac{d}{2}.
\end{eqnarray}
So $\langle(\alpha_{j}^{*}\alpha_{j+\frac{d}{2}})_{j=0}^{d-1}|w_{i}\rangle=0$
where $i=0,\cdots,d-1$ and $(\alpha_{j}^{*}\alpha_{j+\frac{d}{2}})_{j=0}^{d-1}=(\alpha_{0}^{*}\alpha_{\frac{d}{2}},\alpha_{1}^{*}\alpha_{1+\frac{d}{2}},\ldots,\alpha_{d-1}^{*}\alpha_{\frac{d}{2}-1})$.
Therefore $(\alpha_{j}^{*}\alpha_{j+\frac{d}{2}})_{j=0}^{d-1}$ is a zero vector since $\{|w_{i}\rangle\}_{i=0}^{d-1}$ is a base.
This also ensures that there is at least $\frac{d}{2}$ $\alpha_{j}=0$.
The following proof is similar to that of Lemma \ref{lem3.1}, we omit it here.
\end{proof}

\subsection{LOCC indistinguishable set of $3+\lceil\frac{d}{4}\rceil$ GBSs}

\begin{theorem}\label{th4.1}
Let $m$ be a positive integer,
if (1) $d=4m$ and $S^{d}=\{(0,2i-1)\}_{i=1}^{\frac{d}{4}}\cup\{(0,0),(\frac{d}{2},0),(\frac{d}{2},\frac{d}{2}+1)\}$,
or (2) $d=4m+2$ and $S^{d}=\{(0,2i-1)\}_{i=1}^{\lceil\frac{d}{4}\rceil}\cup\{(0,0),(\frac{d}{2},0),(\frac{d}{2},\frac{d}{2})\}$,
then $S^{d}$ is 1-LOCC indistinguishable and $f_{GBS}(d)\leq |S^{d}|= 3+\lceil\frac{d}{4}\rceil.$
In particular, $f_{GBS}(4)\leq 4$, $f_{GBS}(6)\leq 5$ and $f_{GBS}(8)\leq 5$.
\end{theorem}
When $d$ is even and $d\geq 4$, the cardinality of $S^{d}$ ($|S^{d}|=3+\lfloor\frac{d}{4}\rfloor$) in Theorem \ref{th4.1}
is smaller than the result in Zhang et al. \cite{zhang2015pra} ($|S^{d}|=\frac{d+4}{2}$).
See Table \ref{tab4.1} for comparison of $|S^{d}|$ in \cite{zhang2015pra} and Theorem \ref{th4.1}.

\begin{proof}
(1) When $d=4m$,
suppose that $S^{d}$ can be distinguished by 1-LOCC,
then there exists a normalized vector $|\alpha\rangle=\sum_{j=0}^{d-1}\alpha_{j}|j\rangle$
such that $\{U_{m_{j}n_{j}}|\alpha\rangle\}_{j=1}^{l}$ is a orthonormal set.
It means that $\langle\alpha|U_{m_{kj}n_{kj}}|\alpha\rangle=0$, $j\neq k$.
By simple calculation,
\begin{eqnarray}\label{}\nonumber
\Delta U\supseteq\{(0,i)\}_{i=1}^{\frac{d}{2}-1}\cup\{(0,i)\}_{i=\frac{d}{2}+1}^{d-1}\cup\{(\frac{d}{2},i)\}_{i=0}^{d-1}.
\end{eqnarray}
Hence $\langle\alpha|U_{\frac{d}{2},i}|\alpha\rangle=0$, $i=0,\cdots,d-1$;
and $\langle\alpha|U_{0,i}|\alpha\rangle=0$, $i=1,\cdots,\frac{d}{2}-1,\frac{d}{2}+1,\cdots,d-1$.
So $|(\alpha_{j}^{*}\alpha_{j+\frac{d}{2}})_{j=0}^{d-1}\rangle$ is a zero vector
and there exists complex numbers $b_{1},b_{2}$ such that
$|(|\alpha_{j}|^{2})_{j=0}^{d-1}\rangle=b_{1}|w_{0}\rangle +b_{2}|w_{\frac{d}{2}}\rangle=|(b_{1}+(-1)^{j}b_{2})_{j=0}^{d-1}\rangle$.
Therefore $0=|\alpha_{0}^{*}\alpha_{\frac{d}{2}}|^{2}=|\alpha_{0}^{*}|^{2}|\alpha_{\frac{d}{2}}|^{2}=(b_{1}+b_{2})^{2}$
and $0=|\alpha_{1}^{*}\alpha_{1+\frac{d}{2}}|^{2}=(b_{1}-b_{2})^{2}$.
Thus $b_{1}=b_{2}=0$, $|(|\alpha_{j}|^{2})_{j=0}^{d-1}\rangle$ is a zero vector
which contradicts with $|\alpha\rangle\neq 0$.\\
(2) When $d=4m+2$, the assertion follows from Lemma \ref{lem4.1}
and $\Delta U\supseteq\{(0,i)\}_{i=1}^{d-1}\cup\{(\frac{d}{2},i)\}_{i=0}^{d-1}$.
\end{proof}

\subsection{LOCC indistinguishable set of no more than
$2\lceil \sqrt{\frac{d+2}{2}}\rceil+\lceil\frac{\lceil \frac{d+1}{4}\rceil}{\lceil \sqrt{\frac{d+2}{2}}\rceil}\rceil$ GBSs}
In this subsection, assume that  $d\geq 6$ since the case $d=4$ is known.

Similar to Lemma \ref{lem3.2}, the following result follows by Lemma \ref{lem4.1}

\begin{lemma}\label{lem4.2}\nonumber
Let $m$ be a positive integer, $S^{d}(m)=\{(0,i)\}_{i=0}^{m-1}\cup\{(\frac{d}{2},im-1)\}_{i=1}^{\lceil\frac{d+2}{2m}\rceil}
\cup\{(0,\frac{d}{2}-im)\}_{i=0}^{\lceil\frac{\lceil\frac{d+1}{4}\rceil}{m}\rceil-1}$ a set of GBSs,
then $S^{d}(m)$ is 1-LOCC indistinguishable and
$|S^{d}(m)|=m+\lceil\frac{\frac{d+2}{2}}{m}\rceil+\lceil\frac{\lceil\frac{d+1}{4}\rceil}{m}\rceil.$
\end{lemma}
Now we can show an improvement of the case that $d$ is even of the result in Wang et al. \cite{wang2016qip}.
\begin{theorem}\label{th4.2}
If $\lfloor\sqrt{\frac{d+2}{2}}\rfloor^{2}\leq \frac{d+2}{2}\leq \lfloor\sqrt{\frac{d+2}{2}}\rfloor\lceil\sqrt{\frac{d+2}{2}}\rceil$,
then there exists a 1-LOCC indistinguishable set $S^{d}$ of GBSs with $|S^{d}|=\lfloor\sqrt{\frac{d+2}{2}}\rfloor+\lceil\sqrt{\frac{d+2}{2}}\rceil+
\lceil\frac{\lceil\frac{d+1}{4}\rceil}{\lceil\sqrt{\frac{d+2}{2}}\rceil+k_{3}}\rceil$
where $k_{3}=\max\{k: \frac{d+2}{2}\leq (\lfloor\sqrt{\frac{d+2}{2}}\rfloor-k)(\lceil\sqrt{\frac{d+2}{2}}\rceil+k)\}$;
if $\lfloor\sqrt{\frac{d+2}{2}}\rfloor\lceil\sqrt{\frac{d+2}{2}}\rceil<\frac{d+2}{2}\leq \lfloor\sqrt{\frac{d+2}{2}}\rfloor(\lceil\sqrt{\frac{d+2}{2}}\rceil+1)$,
then there exists a 1-LOCC indistinguishable set $S^{d}$ of GBSs with 
$|S^{d}|=\lfloor\sqrt{\frac{d+2}{2}}\rfloor+\lceil\sqrt{\frac{d+2}{2}}\rceil+1+\lceil\frac{\lceil\frac{d+1}{4}\rceil}{\lceil\sqrt{\frac{d+2}{2}}\rceil+1+k_{4}}\rceil$
where $k_{4}=\max\{k: d\leq (\lfloor\sqrt{\frac{d+2}{2}}\rfloor-k)(\lceil\sqrt{\frac{d+2}{2}}\rceil+1+k)\}$.
\end{theorem}
\begin{proof}
If $\lfloor\sqrt{\frac{d+2}{2}}\rfloor^{2}\leq \frac{d+2}{2}\leq \lfloor\sqrt{\frac{d+2}{2}}\rfloor\lceil\sqrt{\frac{d+2}{2}}\rceil$,
by Lemma \ref{lem4.2}, the set $S^{d}(\lceil\sqrt{\frac{d+2}{2}}\rceil+k_{3})$
is 1-LOCC indistinguishable with
$|S^{d}(\lceil\sqrt{\frac{d+2}{2}}\rceil+k_{3})|=
\lfloor\sqrt{\frac{d+2}{2}}\rfloor+\lceil\sqrt{\frac{d+2}{2}}\rceil+\lceil\frac{\lceil\frac{d+1}{4}\rceil}{\lceil\sqrt{\frac{d+2}{2}}\rceil+k_{3}}\rceil.$
Similarly, when $\lfloor\sqrt{\frac{d+2}{2}}\rfloor\lceil\sqrt{\frac{d+2}{2}}\rceil< \frac{d+2}{2}\leq \lfloor\sqrt{\frac{d+2}{2}}\rfloor(\lceil\sqrt{\frac{d+2}{2}}\rceil+1)$,
the set $S^{d}(\lceil\sqrt{\frac{d+2}{2}}\rceil+1+k_{4})$ is 1-LOCC indistinguishable with
$|S^{d}(\lceil\sqrt{\frac{d+2}{2}}\rceil+1+k_{4})|= \lfloor\sqrt{\frac{d+2}{2}}\rfloor+\lceil\sqrt{\frac{d+2}{2}}\rceil+1+\lceil\frac{\lceil\frac{d+1}{4}\rceil}{\lceil\sqrt{\frac{d+2}{2}}\rceil+1+k_{4}}\rceil$.
\end{proof}
When $d\geq 26$, $|S^{d}|$ in Theorem \ref{th4.2} is not more than $|S^{d}|$ in Theorem \ref{th4.1}.
See Table \ref{tab4.2} for comparison of $|S^{d}|$ in \cite{wang2016qip} and Theorem \ref{th4.2}.

\begin{table*}
\caption{\label{tab4.1}Comparison of \cite[Theorem 1]{zhang2015pra} ($|S^{d}|=\frac{d+4}{2}$) and Theorem \ref{th4.1} ($|S^{d}|=3+\lceil\frac{d}{4}\rceil$).}
\footnotesize
\begin{tabular}{c|cccccccccccccccccccccccccc}
\br
$d$ &4&6&8&10&12&14&16&18&20&26$^{c}$&30&40&50&60&70&80&90&100\\
\hline
$|S^{d}|$ in \cite{zhang2015pra}&4&5&6&7&8&9&10&11&12&15&17&22&27&32&37&42&47&52 \\
Theorem \ref{th4.1}&4&5&5&6&6&7&7&8&8&10&11&13&16&18&20&23&26&28\\
\br
\end{tabular}
\end{table*}

\begin{table*}
\caption{\label{tab4.2}Comparison of \cite[Theorem 1]{wang2016qip} ($|S^{d}|=3\lceil \sqrt{d}\rceil-1$) and Theorem \ref{th4.2}
($|S^{d}|=\lfloor\sqrt{\frac{d+2}{2}}\rfloor+\lceil\sqrt{\frac{d+2}{2}}\rceil+
\lceil\frac{\lceil\frac{d+1}{4}\rceil}{\lceil\sqrt{\frac{d+2}{2}}\rceil+k_{3}}\rceil$,
$|S^{d}|=\lfloor\sqrt{\frac{d+2}{2}}\rfloor+\lceil\sqrt{\frac{d+2}{2}}\rceil+1+\lceil\frac{\lceil\frac{d+1}{4}\rceil}{\lceil\sqrt{\frac{d+2}{2}}\rceil+1+k_{4}}\rceil$).}
\footnotesize
\begin{tabular}{c|ccccccccccccccccccccccccc}
\br
$d$  &4&6&8&10&12&14&16&18&20&26$^{c}$&30&40&50&60&70&80&90&100\\
\hline
$|S^{d}|$ in \cite{wang2016qip}&&8&8&11&11&11&11&14&14&17&17&20&23&23&26&26&29&29\\
Theorem \ref{th4.2}&&5&6&6&7&7&7&8&9&10&10&12&13&14&15&16&17&18\\
\br
\end{tabular}
$^{c}$$|S^{d}|$ in Theorem \ref{th4.2} is not more than $|S^{d}|$ in Theorem \ref{th4.1} when $d\geq 26$.
\end{table*}

\section{Conclusion}

In summary, by using linear system and Vandermonde matrix,
we have extended known results on 1-LOCC indistinguishable set of GBSs in $d\otimes d$.
Based on the function $f(d)$ defined by Zhang et al.,
the function $f_{GBS}(d)$ which is the minimum cardinality of 1-LOCC indistinguishable set of GBSs in $d\otimes d$ is introduced.
The extended results imply that, when $d$ is odd, we have
$f_{GBS}(d)\leq\min\{\frac{d+3}{2},\lfloor\frac{d+1}{4}\rfloor+5,2\lceil \sqrt{d}\rceil+\lceil\frac{\lceil \frac{d-1}{4}\rceil}{\lceil \sqrt{d}\rceil}\rceil\};$
when $d$ is even, then
$f_{GBS}(d)\leq\min\{\lceil\frac{d}{4}\rceil+3,2\lceil \sqrt{\frac{d+2}{2}}\rceil+\lceil\frac{\lceil \frac{d+1}{4}\rceil}{\lceil \sqrt{\frac{d+2}{2}}\rceil}\rceil\}.$
In particular,  $f_{GBS}(7)=5$.
The extended results may lead to a better understanding of the nonlocality of maximally entangled states.
It is still an interesting open question whether we can find the exact value of  $f_{GBS}(d)$ for $d= 6$ and $d\geq 8$.

\begin{ack}
This work is supported by NSFC (Grant No. 11301155, 61601171),
Project of Science and Technology Department of Henan Province of China (172102210275, 182102210306),
Foundation of Doctor of Henan Polytechnic University (B2017-48).
\end{ack}

\end{document}